\documentclass[10pt]{llncs}
\usepackage{amsmath, amssymb}
\usepackage{algorithmic}
\usepackage{graphicx,graphics,color,epsfig}
\usepackage{color}

\newcommand{\R}{{\mathbb{R}}}

\newcommand{\offset}{\textbf{offset }}
\newcommand{\vo}{{\vec{o}}}
\newcommand{\vd}{{\vec{d}}}
\newcommand{\lopt}{{\ell^{opt}}}
\newcommand{\lon}{{\ell^{on}}}
\newcommand{\popt}{{p^{opt}}}
\newcommand{\pon}{{p^{on}}}

\newcommand{\ponx}{{p^{on}_x}}

\newcommand{\popty}{{p^{opt}_y}}
\newcommand{\pony}{{p^{on}_y}}

\newcommand{\e}{{\epsilon}}
\renewcommand{\d}{{d(\pon + \vo, \popt)}}

\spnewtheorem{mainth}[theorem]{Main Theorem}{\bfseries}{\itshape}
\spnewtheorem{observation}[theorem]{Observation}{\bfseries}{\itshape}
\spnewtheorem{clm}[theorem]{Claim}{\bfseries}{\itshape}

\title{On the Continuous CNN Problem}
\author{John Augustine\thanks{Work done in part while at Tata Research Development and Design Centre, Pune, India.} \and Nick Gravin}
\institute{School of Physical and Mathematical Sciences\\ Nanyang Technological University\\Singapore 637371.\\ \email{jea@ics.uci.edu, ngravin@gmail.com}}
\begin{document}
\maketitle

\begin{abstract}

In the (discrete) CNN problem, online requests appear as points in $\mathbb{R}^2$. Each request must be served before the next one is revealed. We have a server that can serve a request  simply by aligning either its $x$ or $y$ coordinate with the request. The goal of the online algorithm is to minimize the total $L_1$ distance traveled  by the server to serve all the requests. The best known competitive ratio for the discrete version is 879 (due to Sitters and Stougie). 

We study the continuous version, in which, the request can move continuously in $\mathbb{R}^2$ and the server must continuously serve the request. A simple adversarial argument shows that the lower bound on the competitive ratio of any online algorithm for the continuous CNN problem is 3. Our main contribution is an online algorithm with competitive ratio $3+2 \sqrt{3} \approx 6.464$. Our analysis is tight. The continuous version generalizes  the discrete orthogonal  CNN problem, in which every request must be $x$ or $y$ aligned with the previous request. Therefore, Our result improves upon the previous best competitive ratio of 9 (due to Iwama and Yonezawa).

\end{abstract}

\section{Introduction}

The $k$-server problem has been  influential in the development of online algorithms~\cite{DBLP:journals/csr/Koutsoupias09}. We have $k$ servers that can move around a metric space. Requests arrive in an online manner on various locations in the metric space. After each request arrives, one of the $k$ servers must move to the request location. The online algorithm must make this decision without any knowledge of the future requests.  The objective is to minimize the sum of the distances traveled by the $k$ servers.

A natural variant of the $k$-server problem, the (discrete) CNN problem, was introduced by Koutsoupias and Taylor~\cite{KoutsoupiasT04}. The name derives from the following illustrative example: consider a sequence of newsworthy events that occur in street intersections in Manhattan. A CNN news crew must cover these events with minimal movement. Since they have powerful zoom lenses, they only need to be at some point on either one of the two cross streets. More formally, we are given a sequence of requests as points from $\R^2$ that appear in an online manner. We have one server that can move around in $\R^2$.
%We say that two points are $x$-aligned (respectively, $y$-aligned) if they share the same $x$ (respectively, $y$) coordinate. Also, we say that two points are aligned if they are either $x$-aligned or $y$-aligned.
To serve a request, the server must merely align itself to the $x$ or $y$ coordinate of the request. The objective is to minimize the total distance traveled by the server in $L_1$ norm. 

There is an equivalent  alternative definition that is also used in literature in which, instead of a single server that can move in 2D,  we have two independent servers with one restricted to move along the $x$-axis, while the other is restricted to move along the $y$-axis. Given an online request at $(a,b)$, either the $x$-axis server must move to $x=a$ or the $y$-axis server must move to $y=b$. The objective is to minimize the sum of distances moved by either servers. Notice that the two independent servers in different dimensions are equivalent to a single server that can move around in both dimensions. For this reason, the CNN problem is also called sum of two $1$-server problems~\cite{KoutsoupiasT04}.

\begin{figure}
  \begin{centering}
  \includegraphics[width=2in]{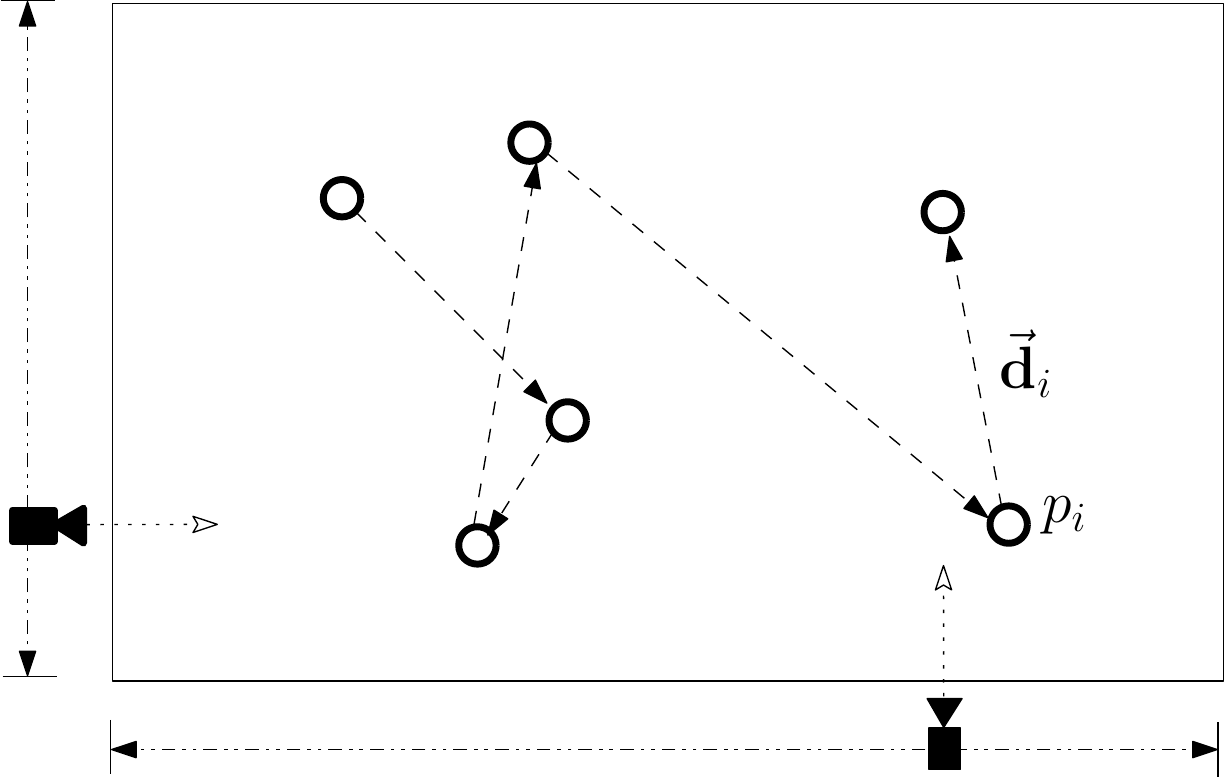}\\
  \caption{Illustration for the two server definition of the continuous CNN problem.} \label{fig:soccer}
    \end{centering}
\end{figure}
We introduce the continuous version of the CNN problem. We use the alternative two server definition to illustrate the continuous version. Consider the problem of covering the activities of a soccer match (see Figure~\ref{fig:soccer}). For the sake of simplicity in our illustration, let us have two  cameras on rails, one along the length (i.e., the $x$-axis server) and the other along the breadth (the $y$-axis server) of the field. Their orientations are fixed perpendicular to the direction of movement (of course, pointing into the field). As the ball is kicked around, at least one of the two cameras must track the ball continuously. Informally, the input is a request point moving along a continuous trajectory that is revealed in an online manner and a server must {\em continuously} align itself to the $x$ or $y$ coordinate of the request. 

We say that two points are $x$-aligned (respectively, $y$-aligned) if they share the same $x$ (respectively, $y$) coordinate. Also, we say that two points are aligned if they are either $x$-aligned or $y$-aligned. We are now ready to formally define the continuous CNN problem. For this formal definition (and for the rest of the paper) 
%we use the single server definition, i.e,  
we have a single server that can move around in 2D 
space. Our input is an online sequence of pairs $r_i = (p_i, \vd_i)$, where $p_i$ is a point on $p_{i-1} + t \vd_{i-1}$, $t \ge 0$, and $\vd_i$ is a unit vector in some arbitrary direction. (In the soccer illustration, $p_i$ is the point on the previous trajectory where the ball is intercepted and $\vd_i$ is the new direction in which it is kicked.) Without loss of generality, the first point is assumed to be the origin. The server also starts at the origin. When an input pair $(p_i, \vd_i)$ is revealed, the server and $p_i$ are already aligned. The online algorithm must then commit to a continuous trajectory $T_i(t)$ of the server parameterized by $t$ such that for all $t\ge0$, $T_i(t)$ is aligned with $p_{i} + t \vd_{i}$.  After the online algorithm commits, the next request $(p_{i+1}, \vd_{i+1})$ arrives, the online server moves to the point on $T_i$ that aligns with $p_{i+1}$ along the trajectory $T_i$. The objective is to minimize the total distance traveled by the server in $L_1$ norm. 

%The continuous CNN problem is fundamentally different from the discrete version. To see this, consider the case where the server is at the origin. Suppose the request first appears at $(1,0)$ and then at, say, $(0,1)$. In the discrete case, the server can stay unmoved, but in the continuous case, the server has to align with the request as it moves from $(1,0)$ to $(0,1)$. 

%\noindent {\bf History of CNN problems: } 
\paragraph{History of CNN problems:} The discrete version of the CNN problem was discussed in several conferences and seminars in the late 1990s without any breakthroughs\cite{Chrobak03,IwamaY04}. It was formally introduced by Koutsoupias and Taylor~\cite{KoutsoupiasT04}\footnote{Conference version appeared in STACS 2000.}. They conjectured that this problem has a competitive algorithm along with a lower bound of $6 + \sqrt{17}$ on the competitive ratio of any deterministic online algorithm. Their conjecture was proved affirmatively in \cite{DBLP:conf/icalp/SittersSP03} by Sitters, Stougie, and de Paepe, albeit, with an algorithm that was $10^5$-competitive. For a fascinating discussion of the prevailing understanding of this problem in 2003, see~\cite{Chrobak03}. Eventually, Sitters and Stougie~\cite{SittersS06} made further improvements and provided a $879$-competitive algorithm. In fact, their work focussed on the {\em generalized $k$-server problem} which can be characterized as the sum of several $1$-server problems on arbitrary metric spaces. The orthogonal CNN problem was introduced by Iwama and Yonezawa~\cite{IwamaY04}. Each request (except the first one) must either share the $x$ coordinate or the $y$ coordinate with the previous request. With this restriction, they were able to improve the competitive ratio dramatically to $9$.

%\noindent {\bf Our Contribution:} 
\paragraph{Our Contribution:} We focus on the continuous CNN problem, which is a generalization of the orthogonal CNN problem. We formalize this in the following Claim (with proof deferred to the Appendix).
\begin{clm}\label{clm:general}
Any $c$-competitive algorithm $\mathcal{A}$ for the continuous CNN problem can be applied to the orthogonal CNN problem in a manner that preserves the competitive ratio.
\end{clm}
A typical adversarial argument gives us the following lower bound on the competitive ratio.
 \begin{clm}\label{clm:lb}
If there is a $c$-competitive algorithm for the continuous CNN problem, then $c \ge 3$ even on a unit square.
\end{clm}
In Section~\ref{sec:prelim}, we introduce a simpler problem called the unit CNN problem and prove a lower bound of 3 on its competitive ratio in Theorem~\ref{thm:lb}. The proof of Theorem~\ref{thm:lb} can be easily adapted for Claim~\ref{clm:lb}.
%Note that Claim~\ref{clm:lb} follows immediately from Theorem~\ref{thm:lb} and . 

We now provide an example that informally illustrates how we get a lower bound of 3 on the competitive ratio of the continuous CNN problem; see Figure~\ref{fig:lb}. Consider the unit square with both the optimal offline server and the online server at the top-left corner. In this adversarial example, the request moves to the bottom-right corner so that the online server is forced to choose between either a clockwise or counter-clockwise direction. Assume, without loss of generality, that it chooses the clockwise direction and moves to the top-right. The offline server, however, makes a single move down to the bottom-left. Suppose now the request moves around repeatedly in the left and bottom edges of the unit square, i.e., it makes a repeated ``L" shaped move. Clearly, the  offline server is already at a ``sweet spot'' and therefore stays unmoved. The online server, however, must correct its position and move to the sweet spot to offset its disadvantage. Notice that the online server moved three units of distance while the optimal offline server just needed just one.
\begin{figure}
  \begin{centering}
  \includegraphics[width=2in]{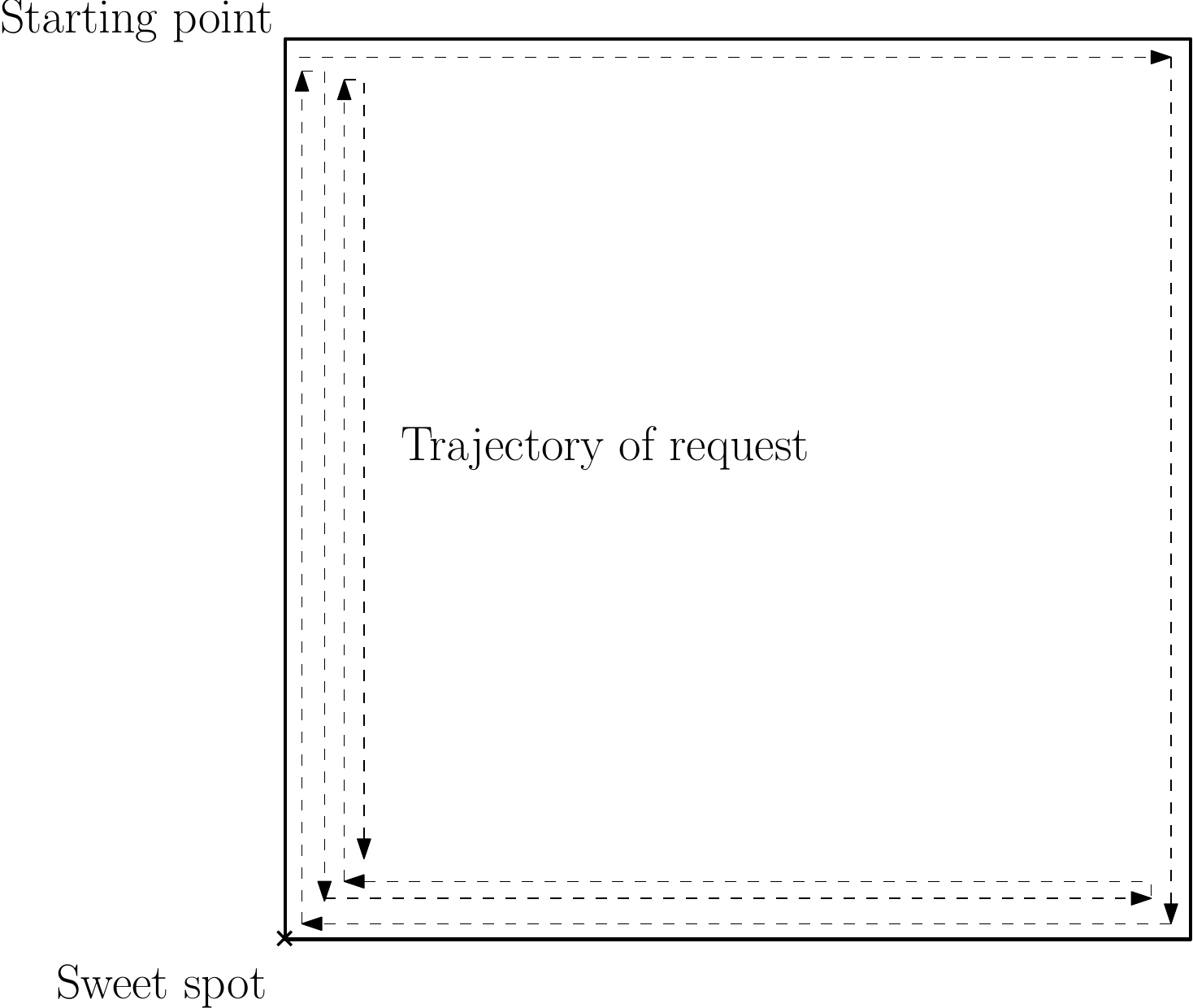} \hspace{0.5in}
  \includegraphics[width=2in]{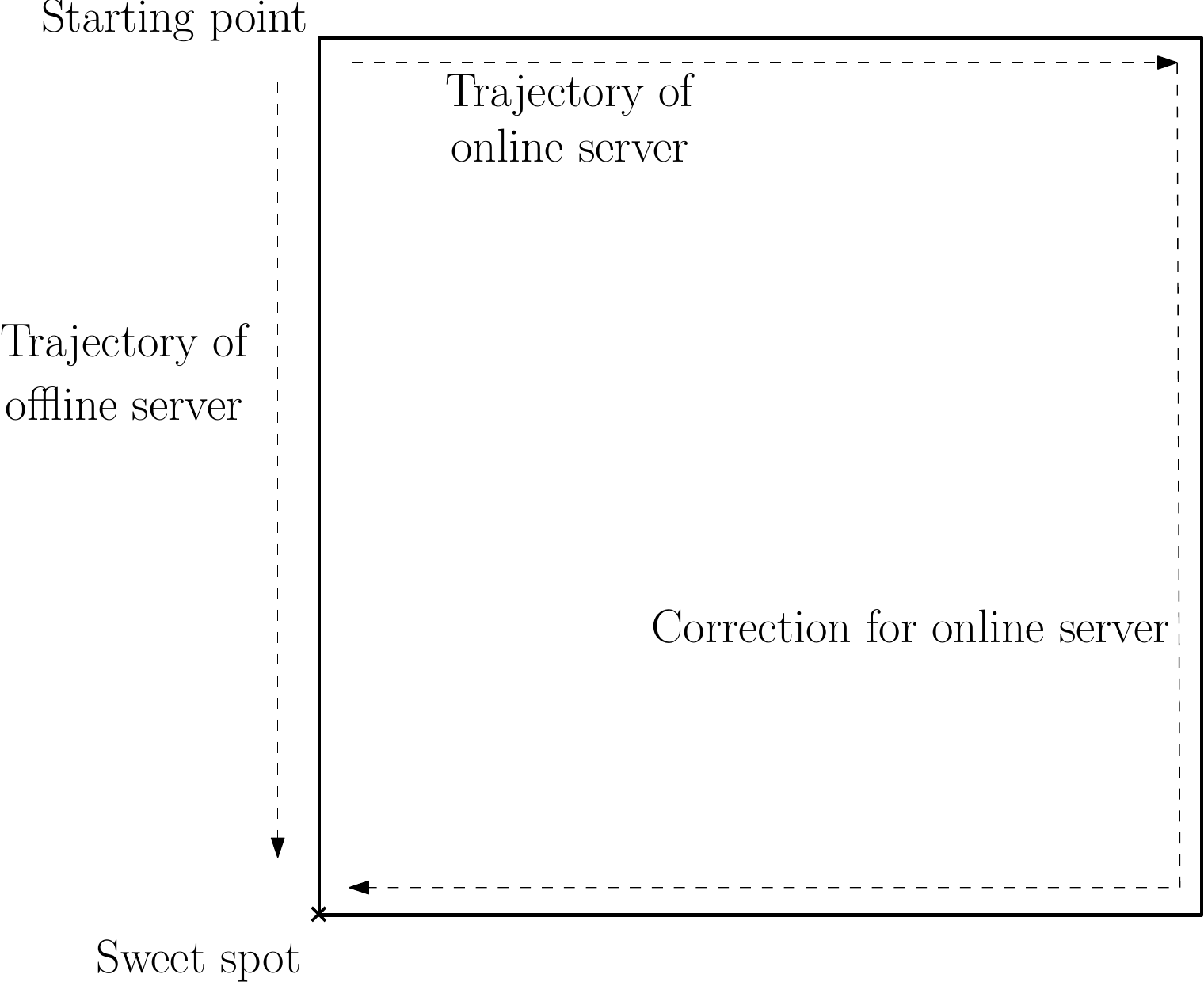}\\
  \caption{The figure on the left shows the request trajectory. The figure on the right shows the trajectory of online and offline servers.
  %The request moves in a clockwise manner prompting the online server to follow while the offline optimal takes a shortcut. The online server must then correct itself.
  } \label{fig:lb}
    \end{centering}
\end{figure}

The significant contribution of our paper is an online algorithm for the continuous CNN problem with a competitive ratio of  $3+2\sqrt{3} = 6.464$. 
%The intuition behind this algorithm is built on the proof idea of Claim~\ref{clm:lb} where the online server homes into the sweet spot.
In light of Claim~\ref{clm:general}, our result improves upon the $9$-competitive algorithm for the orthogonal CNN problem~\cite{IwamaY04}.
Our algorithm alternates between two phases, namely, the {\em bishop phase} and the {\em rook phase.} Hence, we call it the Bishop-Rook algorithm or just the BR algorithm. 
%In the bishop phase, as the name implies, the online server moves in a diagonal manner and during the rook phase, it moves orthogonal to the axes. At the end of each bishop phase, the server constructs an {\em offset} vector that points to the sweet spot of the offline server. During the rook phase, it corrects its position based on the offset vector. 
Our analysis using a non-decreasing potential function is non-trivial.
Finally, we show that our analysis is tight by constructing input instances for which the competitive ratio is realized.

\paragraph{Organization of the paper:} 
In Section~\ref{sec:prelim}, we introduce a simplified problem called the unit CNN problem. We prove a lower bound of 3 on the competitive ratio of algorithms for the unit CNN problem, from which, the lower bound of the continuous CNN problem follows quite immediately. 
In Section~\ref{sec:BR} we present the {\tt BR} Algorithm for the continuous CNN problem. We analyze the {\tt BR} algorithm in Section~\ref{sec:analysis} and show that it has a competitive ratio of $(3 + 2 \sqrt{3}) \approx 6.464$.
%We end with some concluding remarks in Section~\ref{sec:conclusion}.

\section{Preliminaries: Lower Bound}

\label{sec:prelim}
The lower bound for the competitive ratio of the continuous CNN problem (see Claim~\ref{clm:lb}) is 3. As it turns out, this lower bound shows up in a much simpler problem that we call the {\em unit CNN problem}. We provide the lower bound proof for the unit CNN problem (Theorem~\ref{thm:lb}) which easily encompasses Claim~\ref{clm:lb}.  
%For the sake of completion, we state that there exists a 4-competitive algorithm for the unit CNN problem, but defer the details to the full version of the paper.
For the sake of completeness, we state the upper bound for the unit CNN problem, but defer all details and proofs to the Appendix.

In the unit CNN problem, a sequence of requests, $(r_1, r_2, \ldots )$, appear online as points in $\R^2$. We have a server that must serve each request $r_i$ by moving to a location that is aligned vertically or horizontally with the request. Each time the server moves either horizontally or vertically, it has to pay \$1 (regardless of the distance moved). If the move has both a horizontal and a vertical component, it pays \$2. The objective of the server is to minimize the total payment. We also consider the restricted version called the {\em unit orthogonal CNN problem} in which any pair of consecutive requests must either share an $x$ or $y$ coordinate.

We begin with a simple observation that we state without proof because, in essence, it has been noticed before by Iwama and Yonezawa~\cite{IwamaY04}. Define a sequence of moves by any algorithm to be {\em frugal} if, for each move in the sequence, the payment is never more than the minimum required to serve the current outstanding request. Therefore, by definition, no frugal sequence of moves will include a \$2 move.
\begin{observation}\label{obs:frugal}
Given any sequence of moves that pays \$d to serve the input sequence of requests, there is a frugal sequence that also pays at most \$d to serve all the requests.
\end{observation}

\begin{theorem}\label{thm:lb}
If there is a $c$-competitive algorithm for the unit CNN problem, then $c \ge 3$ even when the requests arrive in an orthogonal manner.
\end{theorem}
\begin{proof}
For this proof, we limit ourselves to input instances in which the requests can only appear in the vertices of the unit square with the origin at its bottom left.
Furthermore, we restrict the requests to appear in an orthogonal manner, i.e., each request (except the first one) shares either an $x$ or $y$ coordinate with the previous request.

Given Observation~\ref{obs:frugal}, we can assume that the server will not move if it is already aligned to the request.
Notice  that an adversary can force the online
algorithm to move at each step. We keep server one step away from each
request $r_i$ by placing next request $r_{i+1}$ diagonally opposite to
the current position of the server.

Given such a sequence of requests, an offline algorithm splits it into
consecutive triples. The offline server moves at most once to 
service each triple $r_{i-1}$, $r_i$ and $r_{i+1}$. 
% just by hopping in one move to $r_i$ (if possible). 
If possible, the offline server moves to $r_i$ in one step --- it can service $r_{i-1}$, $r_i$ and $r_{i+1}$ from $r_i$.
Such a one step move to $r_i$ is not possible only when the server and $r_i$ are diagonally opposite each other, so the server must be within 1 hop of both $r_{i-1}$ and $r_{i+1}$. It simply serves $r_{i-1}$ from the current position and then hops to $r_{i+1}$ to serve both $r_i$ and $r_{i+1}$.

So, for every three steps of the online algorithm, the offline server requires at most 1 move, thereby completing the proof. \qed
\end{proof}

\begin{clm} \label{clm:4c}
There exists a 4-competitive algorithm for the unit CNN problem
\end{clm}

\begin{clm}\label{clm:3c}
There exists a 3-competitive algorithm for the orthogonal unit CNN problem.
\end{clm}

\section{The {\tt BR} Algorithm for the Continuous CNN Problem}

\label{sec:BR}

We now turn our attention to the main problem that we address in this paper --- the continuous CNN problem.
Recall that we formally defined the input as  an online sequence of pairs $(p_i, \vd_i)$. 
Informally, we treat the request as a point starting at the origin and moving to each subsequent $p_i$ in  straight line segments whose direction is given by the vector $\vd_i$. So we use the term {\em request trajectory} to refer to the path traversed by the request. The server's trajectory must stay aligned with request trajectory at all times. In this section, we describe the Bishop-Rook algorithm or just the {\tt BR} algorithm that alternates between two phases, namely, the Bishop phase and the Rook phase. As the name implies, the server moves diagonally during the Bishop phase. In the Rook phase, we treat the horizontal and vertical components of the server separately, leading to movements that mimic  Rooks in Chess. The algorithm switches between the phases when appropriate conditions (described subsequently for each phase) are met.
%In the next section, we show that the {\tt BR} algorithm is $(3+2 \sqrt{3})$-competitive for the continuous CNN problem.

The key intuition behind the algorithm is the following. Suppose the offline server manages to get to a ``sweet spot" from which it can align with the request trajectory with little or no movement. Then, the online server also must home into that spot. Iwama and Yonezawa~\cite{IwamaY04} also exploit this idea. They get closer to a potential sweet spot using ``L" shaped moves --- hence, one can call it the Knight algorithm. 
%The online server does not explicitly correct its position in the Knight algorithm. 
%In the {\tt BR} algorithm, we try to locate the sweet spot using a number of previous request positions. 
To achieve this homing effect in the {\tt BR} algorithm, we define an \offset vector at the end of the bishop phase that, when added to the online server's position, will point to our candidate sweet spot. In the rook phase, we use the \offset vector to guide the online server to the sweet spot.

\begin{figure}
\begin{minipage}{4cm}
\center
\include{phase1}
\caption{Bishop phase}
\label{Phase1}
\end{minipage} \hfill
\begin{minipage}{5cm}
\center
\include{phase2Basic}
\caption{Rook phase without \offset update}
\label{Phase2Basic}
\end{minipage} \hfill
\begin{minipage}{3.5cm}
\center
\include{phase2}
\caption{Rook phase showing \offset update}
\label{Phase2}
\end{minipage} \hfill
\end{figure}

%\vspace{-0.3in}
\subsubsection{Bishop Phase:}
During the bishop phase, as the name implies, the server moves diagonally making a $45^\circ$ angle with the axes. 
%The \offset  vector $\vo$ is set to 0 in the beginning of the bishop phase.  
Without loss of generality, let the point $p_{i}$  be at $(0,0)$ and the online server be on the non-negative part of $y$-axis at $(0,h)$, so $h\ge0$; see Figure~\ref{Phase1}. Throughout the bishop phase, the server moves in a manner that maintains $x$-alignment with the request trajectory. Notice that this defines the $x$ component of the server movement. To ensure the diagonal movement of the bishop phase, the server also moves in the $-y$ direction. For every maximal $\delta x$ that the server moves in either the $+x$ or $-x$ direction, the server simultaneously moves a distance $|\delta x|$ in the $-y$ direction. If (and when) the position of server and request trajectory coincide, we terminate the bishop phase and switch to the rook phase. Let $(s_x, s_y)$ be the coordinates of the point at which they coincide. Then, the \offset vector $\vo = -s_x \vec{x}$, where $\vec{x}$ is the unit vector in the positive $x$ direction.

%The following iterative procedure is repeated till the server trajectory meets the request trajectory. At that point, we switch to the rook phase.
%
%Without loss of generality, let the point $p_{i}$  be at $(0,0)$ and the online server be on the non-negative part of $y$-axis at $(0,h)$, so $h\ge0$. Consider the lower envelope of lines $y = h -x$ and $y = h+x$. The online server moves along this lower envelope and stays $x$-aligned with the request. The incremental offset for this iteration is the vector $(-\ponx, 0)$ where $\ponx$ is the $x$-coordinate of the server.
%\begin{description}
%	\item[Case: a new pair arrives. ] We add the incremental offset to $\vo$ and repeat the procedure.
%	\item[Case: request trajectory touches lower envelope. ] We add the incremental offset to $\vo$ and switch to Phase 2.
%\end{description}

\subsubsection{Rook Phase:}
At the beginning of the rook phase, positions of server and request trajectory coincide. Without loss of generality we assume that \offset is in the $-x$ direction. We maintain two invariants throughout the rook phase. 
%The goal is to minimize the $L_1$ distance traveled by the server while maintaining the two invariants.
However, in so doing, we are judicious with the $L_1$ distance traveled by the server.
\begin{description}
  \item[$y$-alignment:] The server and request trajectory are always $y$-aligned. This fully defines the movement of the server along the $y$ direction because the server maintains the same $y$ coordinate as that of the request.
  \item[$x$ coordinate inequality:] The $x$ coordinate of the server is always less than or equal to the $x$ coordinate of the request trajectory. This invariant is more subtle. When the $x$ coordinate of the request trajectory is strictly greater than that of the server, the server's $x$ coordinate stays unchanged --- this is to ensure that we are judicious with the $L_1$ distance traveled. When the $x$ coordinates coincide and the request trajectory is moving in the $-x$ direction, then the server moves along with the request trajectory.
\end{description}

During the rook phase, the \offset vector $\vo$ decreases whenever the server moves. The rate of decrease depends on the horizontal and vertical components of the movement. The rate at which $|\vo|$ decreases is given by:
\[
|\vo| \leftarrow   \left\{
\begin{array}{l l}
  |\vo| - (1 +\sqrt{3}) |t| & \quad \text{if server and request move a distance $t$ vertically}\\
  |\vo| & \quad \text{if request moves but server does not}\\
  |\vo| - t & \quad \text{if server and request move a distance $t$ horizontally}\\
\end{array} \right.
\]

%As usual, we conveniently insert a request point just when $|\vo|$ reaches 0 and switch to phase 1. 
When $|\vo|$ reaches 0, we switch to the bishop phase.
Fig.~\ref{Phase2Basic} depicts the working of the rook phase, but does not show the change in \offset. Fig.~\ref{Phase2} shows how the \offset shrinks as the phase progresses.

%We now comment on applying the offset algorithm to the non-orthogonal continuous CNN problem. Notice that each $s_i$ is no longer guaranteed to be axis parallel. We overcome this difficulty by considering the $x$ and $y$ components of each $s_i$ separately.

%\begin{figure}
%\begin{minipage}{5cm}
%\center
%\include{Phase1}
%\caption{Phase 1}
%\label{Phase1}
%\end{minipage} \hfill
%\begin{minipage}{6cm}
%\center
%\include{Phase2}
%\caption{Phase 2}
%\label{Phase2}
%\end{minipage} \hfill
%\end{figure}

\section{Analysis of the {\tt BR} Algorithm}

\label{sec:analysis}

To simplify the analysis, we assume that we are working on an instance of the continuous {\em orthogonal} CNN problem, i.e, all the direction vectors $\vd_i$ are orthogonal with respect to the axes. This does not affect our analysis because any straight line of arbitrary angle can be approximated by a series of infinitesimally small $x$ and $y$ components.

Before we proceed with the analysis, we make a simple observation that allows us to insert artificial points into the input sequence. Suppose we are given a sequence of input requests $I=((p_1, \vd_1), \ldots, (p_i, \vd_i), (p_{i+1}, \vd_{i+1}), \ldots)$. Consider the sequence $I'=(p_1, \vd_1), \ldots, (p_i, \vd_i), (p_i', \vd_i), (p_{i+1}, \vd_{i+1}), \ldots$, where $p_i'$ lies on the line segment between $p_i$ and $p_{i+1}$. Then any server trajectory for serving the request sequence $I$ will also serve $I'$ and vice versa.

Our analysis uses a potential function that is non-decreasing throughout the execution of the algorithm. We define a cycle to be the combination of a bishop phase and the subsequent rook phase. Recall that at the start of a cycle, the \offset is 0. We re-orient our view such that the next outstanding request is at the origin and the online server is at $(0,h)$, where $h \ge 0$. When re-orienting our view, we ensure that the potential remains unchanged. This is shown formally in Remark~\ref{rem:reorient}. The potential function $\Phi$ is a function of the $\offset$ and the parameters defined as follows:
\begin{description}
	\item[$\lopt$ and $\lon$] are the distances traveled by the optimal offline server and the online server, respectively,
	\item[$\popt$ and $\pon$] are the positions of the optimal offline server and the online server, respectively.
\end{description}
The potential function is given by
\begin{equation}
	\Phi = (3 + 2 \sqrt{3}) \lopt - 3 d(\pon + \vo, \popt) - \lon - |\vo| + f(|\vo|, \popt, \pon),
	\label{eq:pot}
\end{equation}
where $d(p, q)$ is the $L_1$ distance between points $p$ and $q$. To define $f$, we first define $h = \popty - \pony$, where $\popty$ and $\pony$ are the $y$ coordinates of $\popt$ and $\pon$. Now,

\[
f(\vo, \popt, \pon) =   \left\{
\begin{array}{l l}
  0 & \quad \text{if $h \le 0$}\\
  (6-2\sqrt{3})h & \quad \text{if $0 \le h \le |\vo|$}\\
  (6-2\sqrt{3})|\vo| & \quad \text{if $|\vo| \le h$}\\
\end{array} \right.
\]

\begin{theorem}\label{thm:main}
$\Phi$ is non-decreasing throughout the execution of the {\tt BR} algorithm and this implies a competitive ratio of $(3 + 2 \sqrt{3})$.
\end{theorem}
We first provide a series of  lemmas that lead to the proof of Theorem~\ref{thm:main}.
\begin{lemma}\label{lem:basic}
If the online server stays still, $\Phi$ does not decrease.
\end{lemma}
\begin{proof}
Note that $\vo$ remains unchanged when the online server stays still. Also, the optimal server either (i) does not move, (ii) moves horizontally (arbitrary distance) or (iii) moves vertically the same distance that the request moves. In all three cases, $\Phi$ does not decrease.
\qed\end{proof}

\begin{corollary}\label{cor:one-ver}
From Lemma~\ref{lem:basic}, it follows that, in the bishop phase,  $\Phi$ does not decrease when request moves vertically.
\end{corollary}

We define the offset halfplane to be the halfplane $x \le \ponx$. Naturally, its complement is $x > \ponx$. Since $\ponx$ can change as the online server moves, the offset halfplane also changes accordingly.

\begin{corollary}\label{cor:two-hor-right}
From Lemma~\ref{lem:basic}, it follows that, in the rook phase,  $\Phi$ does not decrease when request moves horizontally in the complement of the offset halfplane.
\end{corollary}

\begin{remark}\label{rem:reorient}
At the start of each cycle, the axes of the euclidean plane can be redrawn (orthogonally) without changing $\Phi$.
\end{remark}
\begin{proof}
At the start of each cycle, \offset is 0. Therefore, only the first three terms of Equation~\ref{eq:pot} are non-zero. Those three terms do not change if the axes are redrawn orthogonal to the previous axes.\qed
\end{proof}

In the rest of the lemmas, since we can insert new points into the request sequence,  we show that $\Phi$ does not decrease for small $\e$ moves of the request in the direction specified.

\begin{lemma} \label{lem:one-hor}
In the bishop phase, $\Phi$ does not decrease when the request moves a distance $\e$ in the horizontal direction.
\end{lemma}
\begin{proof}
We treat this proof in cases based on the behavior of the optimal offline algorithm.
\begin{description}
	\item[Case: $\popt$ is unchanged. ] This is only possible if $\popt$ and request are $y$-aligned. $\lopt$ is unchanged. $\lon$ increased by $2\e$. $|\vo|$ has changed by at most $\e$. $f=0$ because $h \le 0$. If the request moves in the same direction as $\vo$, then $|\vo|$ decreases by $\e$. $\d$ decreased by $\e$.  Overall, $\Phi$ does not decrease (see Fig.~\ref{l2c1}).
	\item[Case: $\popt$ moves vertically and aligns with request.] This is a composition of Lemma~\ref{lem:basic} and the previous case (see Fig.~\ref{l2c2}).
	\item[Case: $\popty \le \pony$ and $\popt$ and request are $x$-aligned for the duration of the move.]
	$\lopt$ increases by $\e$. If request moves in the same direction as $\vo$, then $|\vo|$ decreases by $\e$ and $\d$ decreases by $2\e$, otherwise, $|\vo|$ increases by $\e$ and $\d$ is unchanged. $\lon$ increases by $2\e$. $f$ remains at 0. Therefore, $\Phi$ does not decrease (see Fig.~\ref{l2c3}).
	\item[Case: $\popty \ge \pony$ and $\popt$ and request are $x$-aligned  for the duration of the move.] $\lopt$ increases by $\e$. If request moves in the same direction as $\vo$, then $|\vo|$ decreases by $\e$ and $\d$ remains unchanged. Otherwise, $|\vo|$ increases by $\e$ and $\d$ increases by $2\e$. $\lon$ increases by $2\e$. $h$ in $f$ increased by $\e$ (see Fig.~\ref{l2c4}).
	
	The easy case is when $|\vo|$ decreases. We assume that either $|\vo| \ge h$ or $|\vo| \le h$. Otherwise, we can insert a request when the change happens. With either option, the change in $f$ term is positive and since $|\vo|$ decreases, one can work out that $\Phi$ increases.
	
	When  $|\vo|$ increases, the analysis tightens. The $f$ term increases by $(6 - 2 \sqrt{3}) \e$ because $|\vo|$ and $h$  also increase by $\e$. Therefore, $\Delta \Phi = (3+2 \sqrt{3}) \e - 6 \e - 2 \e - \e + (6 - 2 \sqrt{3}) \e = 0$.

	\item[Case: $\popt$ and request are $x$-aligned  for the duration of the move.] In this case, we are not restricting the relative locations of $\popt$ and $\pon$. In particular,  $\pon \ge \popt$ first, then after some point, the inequality is interchanged. If we insert a request at that point, then, this case breaks into the previous two cases.
\end{description}\qed
\begin{figure}[ht]
\begin{minipage}{4cm}
\center
\include{l2c1}
\caption{Case: $\popt$ is unchanged}
\label{l2c1}
\end{minipage} \hfill
\begin{minipage}{4cm}
\center
\include{l2c2}
\caption{Case: $\popt$ moves vertically}
\label{l2c2}
\end{minipage} \hfill%\\
\begin{minipage}{4cm}
\center
\include{l2c3}
\caption{Case: $\popty \le \pony$}
\label{l2c3}
\end{minipage} \hfill
\begin{minipage}{4cm}
\center
\include{l2c4}
\caption{Case: $\popty \ge \pony$}
\label{l2c4}
\end{minipage} \hfill
\end{figure}
\end{proof}

\begin{lemma}\label{lem:two-hor-left}
In the rook phase,  $\Phi$ does not decrease when request moves horizontally into the offset halfplane.
\end{lemma}
\begin{proof}
As the request moves a distance $\e$, the online server goes with it. (So, the request does not enter the offset halfplane, but rather pushes it by a distance $\e$.) Therefore, $|\vo|$ decreases by $\e$.
\begin{description}
	\item[Case: $\popt$ stays still. ] Clearly, $\popt$ is $y$-aligned with the request. $\lopt$ and $\d$ are unchanged, but $\lon$ increases by $\e$. Since $\popt$ and $\pon$ are $y$-aligned, $f = 0$. Recall that $|\vo|$ decreases by $\e$. Therefore, $\Phi$ is unchanged (see Fig.~\ref{l3c1}).
	\item[Case: $\popt$ makes a vertical move after which, $\popt$ and request are $y$-aligned. ] We can assume that $\popt$ made the jump first before $\pon$ moved along with the request. From Lemma~\ref{lem:basic}, $\Phi$ does not decrease when $\popt$ jumped. $\pon$ moving along with the request is handled by the previous case (see Fig.~\ref{l3c2}).
\item[Case: $\popt$ makes an $x$-aligned move.] $\lopt$ and $\lon$ increase by $\e$. $\d$ decreased by $\e$. Since $|\vo|$ decreased by $\e$, $f$ decreases at most by $(6 - 2 \sqrt{3})\e$ (see Fig.~\ref{l3c3}). Therefore,
	\[\Delta(\Phi) \ge (3+2\sqrt{3}) \e + 3\e + \e - \e - (6-2\sqrt{3}) \e \ge 0.\quad\quad\quad\quad\quad\quad\qed
\]
\end{description}\end{proof}
\begin{figure}[ht]
\begin{minipage}{5cm}
\center
\include{l3c1}
\caption{Case: $\popt$ stays still}
\label{l3c1}
\end{minipage} \hfill
\begin{minipage}{5cm}
\center
\include{l3c2}
\caption{Case: $\popt$ moves vertically}
\label{l3c2}
\end{minipage} \hfill
\begin{minipage}{5cm}
\center
\include{l3c3}
\caption{Case: $\popt$ makes an $x$-aligned move}
\label{l3c3}
\end{minipage} \hfill
\end{figure}

%\vspace{-.2in}
\begin{lemma}\label{lem:two-ver}
In the rook phase, $\Phi$ does not decrease when request moves vertically.
\end{lemma}
\begin{proof}
Note that all vertical moves of $\e$ distance in the rook phase decrease $|\vo|$ by $(1 + \sqrt{3}) \e$.
\begin{description}
	\item[Case: $\popty \le \pony$ and $\popt$ is $x$-aligned and therefore does not move.] $\lopt$ is obviously unchanged, but $\lon$ increases by $\e$. $\d$ decreased by at least $(1 + \sqrt{3}) \e - \e = \sqrt{3}\e$. Since $h \le 0$,  $\Delta(f) = 0$. Therefore, $\Delta(\Phi) \ge 3 \sqrt{3} \e + \sqrt{3} \e > 0$ (see Fig.~\ref{l4c1}).
	\item[Case: $\popty \ge \pony$ and $\popt$ is $x$-aligned, so it does not move. $\pon$ and request move up by $\e$.]
	As in the previous case, $\lopt$ remains unchanged, but $\lon$ increases by $\e$. Also, $|\vo|$ decreased by $(1 + \sqrt{3})\e$. $\d$ decreased by $(1 + \sqrt{3}) \e + \e = 2\e+ \sqrt{3}\e$. Both $h$ and $|\vo|$ decreased, so $f$ decreased as well by at most $(1 + \sqrt{3})(6 -2 \sqrt{3}) \e$ (see Fig.~\ref{l4c2}). Therefore,
	
\begin{equation*}
	\Delta(\Phi) \ge 3(2+ \sqrt{3})\e - \e + (1 + \sqrt{3}) \e - (6-2\sqrt{3})(1 + \sqrt{3}) \e \ge 0
\end{equation*}

\item[Case: $\popt$ is still, but $\pon$ and request start below $\popt$, move up and cross over to above $\popt$.]
This is simply a composition of the above two cases, so $\Phi$ does not decrease.
	\item[Case: $\popty \ge \pony$ and $\popt$ is $x$-aligned, so it does not move. $\pon$ and request move down by $\e$.] $\lopt$ is unchanged, but $\lon$ increases by $\e$. $|\vo|$ decreased by $(1 + \sqrt{3})\e$. $\d$ decreased by $(1 + \sqrt{3}) \e - \e = \sqrt{3}\e$. While $h$ increases, $|\vo|$ decreased. Therefore, $f$ might decrease, but at most by $(1 +\sqrt{3})(6-2\sqrt{3}) \e$. Therefore, $\Delta(\Phi) \ge 3\sqrt{3}\e + \sqrt{3} \e - (1 + \sqrt{3})(6-2\sqrt{3}) \e = 0$ (see Fig.~\ref{l4c3}).

	\item[Case: $\popt$ starts out $y$-aligned, but it $x$-aligns itself to the request with a horizontal move.] This case can be viewed as the composition of two parts. $\popt$ moves first and $\Phi$ does not decrease (by Lemma~\ref{lem:basic}). Then, $\popt$ stays still, but request and server move up. This is the previous case. Hence, $\Phi$ does not decrease (see Fig.~\ref{l4c4}).
	\item[Case: $\popt$ moves vertically (up or down) and stays $y$-aligned.] $\lopt$ and $\lon$ increase by $\e$. $|\vo|$ decreases by $(1 + \sqrt{3})\e$, but $\d$ increases by at most $(1 + \sqrt{3})\e$. Finally, $f$ remains at 0. Therefore, $\Delta(\Phi) \ge (3+ 2 \sqrt{3}) \e - 3(1 + \sqrt{3}) \e - \e + \e + \sqrt{3}\e = 0$ (see Fig.~\ref{l4c5}).\qed
\end{description}\end{proof}
%\vspace{-.2in}
\begin{figure}[ht]
\begin{minipage}{5cm}
\center
\include{l4c1}
\caption{Case: $\popty \le \pony$ and $\popt$ is $x$-aligned}
\label{l4c1}
\end{minipage} \hfill
\begin{minipage}{5cm}
\center
\include{l4c2}
\caption{Case: $\popty \ge \pony$ and $\popt$ is $x$-aligned. $\pon$ moves up}
\label{l4c2}
\end{minipage} \hfill
\begin{minipage}{5cm}
\center
\include{l4c3}
\caption{Case: $\popty \ge \pony$ and $\popt$ is $x$-aligned. $\pon$ moves down}
\label{l4c3}
\end{minipage} \hfill
\begin{minipage}{5cm}
\center
\include{l4c4}
\caption{Case: $\popt$ $x$-aligns with a horizontal move}
\label{l4c4}
\end{minipage} \hfill
\begin{minipage}{5cm}
\center
\include{l4c5}
\caption{Case: $\popt$ moves vertically}
\label{l4c5}
\end{minipage} \hfill
\end{figure}
%\vspace{-.3in}
\begin{proof}[of Theorem~\ref{thm:main}] $\Phi$ started at 0 and, from Lemmas \ref{lem:basic}, \ref{lem:one-hor}, \ref{lem:two-hor-left}, \ref{lem:two-ver}, and Corollary~\ref{cor:two-hor-right}, we know that it only increased. Without loss of generality, we can assume that we terminate at the end of the rook phase, at which point, the $f$ function will evaluate to 0. If we terminate at some other point in the cycle, $f$ might be non-zero. For the purpose of analysis, we can perform a simple trick to bring $f$ to zero without increasing $\lopt$. In particular,  we artificially move the request repeatedly in an ``L" shaped manner with $\popt$ at the corner. $\pon$ will home in on this corner point as well and once it coincides with the corner, $f$ will become zero without incurring any increase in $\lopt$. Since $\Phi = (3 + 2 \sqrt{3}) \lopt - 3 d(\pon + \vo, \popt) - \lon - |\vo| \ge 0$, and $\d$ and $|\vo|$ are non negative, $(3 + 2 \sqrt{3}) \lopt \ge \lon$. \qed
\end{proof}

\begin{remark} The analysis of our algorithm is tight, i.e. there are infinite sequences of requests for which $\lon = (3+2\sqrt{3}) \lopt$. \label{rem:tight}
\end{remark}
%\vspace{-.2in}
\begin{proof}
We provide two different sequence of input. Our first sequence starts with  $\popt$ at the origin. $\pon$ and request are at $(0,1)$ and we are at the beginning of the bishop phase. Then request  moves to $(0,0)$ and then to $(1,0)$.  $\pon$ makes a diagonal move to  $(1,0)$, but $\popt$ does not move. So, $\lon$ has increased by 2, but $\lopt$ stays unchanged. At this moment our algorithm is in the beginning of the rook phase with $\vo = -\mathbf{x}$, where $\mathbf{x}$ is the unit vector along $x$-axis. Request  moves from $(1,0)$ to $(1,\frac{1}{1+\sqrt{3}})$; $\popt$ moves together with it from $(0,0)$ to
	    	    $(0,\frac{1}{1+\sqrt{3}})$. According to our algorithm $\pon$ should follow after the request and at the end of its                move $\vo$ becomes $0$. Thus we get to the bishop phase of our algorithm. $\lopt$ and $\lon$ have both increased by $\frac{1}{1+\sqrt{3}}$. Note that we are back in the situation that we started, so the adversary can repeat this sequence ad infinitum. In each cycle, $\lopt$ has increased by $\frac{1}{1+\sqrt{3}}$ and $\lon$
increased by  $2+\frac{1}{1+\sqrt{3}}$ length in total. Therefore, $\frac{\lon}{\lopt} = (2+\frac{1}{1+\sqrt{3}})(1+\sqrt{3})=3+2\sqrt{3}$.      	 %The second sequence is as follows. Both $\popt$ and $\pon$ be at $(0,1)$ in the beginning of the bishop phase. The request is at the origin. Then request  moves to $(1,0)$ and $\pon$  moves to $(1,0)$ and $\lon$ increases by $2$; $\popt$ moves to $(1,1)$. At this moment our algorithm is in the beginning of the rook phase with $\vo = -\mathbf{x}$. Request moves from $(1,0)$ to $(1,-\frac{1}{1+\sqrt{3}})$; OPT stays still. Server follows after the request and finally reach the $0$ shift. Thus we get to the bishop phase the algorithm. Then request moves to $(1,1)$; $\popt$ and $\pon$ stay. Then, request  moves to  $(-\frac{1}{1+\sqrt{3}},1)$ and then returns back to $(1,1)$; $\popt$ stays still; $\pon$ moves to $(1,1)$ at a  total length $3(1+\frac{1}{1+\sqrt{3}}).$ At this point,  we are again at the beginning of the bishop phase. Then request move at the unit length from $(1,1)$ and we get in the similar situation from which we have started. Again, $\frac{\lon}{\lopt} = 4(1+\frac{1}{1+\sqrt{3}})+1 = 3+2\sqrt{3}$. 
\qed
\end{proof}

In the proof of Remark~\ref{rem:tight}, we provide two tight examples (although the second example is deferred to the Appendix) to indicate how $\Phi$ balances between multiple scenarios. We feel that minor adjustments to $\Phi$ will not reduce the competitive ratio.
%\vspace{-.1in}
%\section{Conclusion and Future Work} \label{sec:conclusion}
%%\vspace{-.1in}
%We have studied CNN problems motivated by exploiting spatial locality in multidimensional data sets. Our main contribution is a 6.464-competitive algorithm to the continuous CNN problem. Our algorithm is quite easy to implement. However, the analysis is admittedly quite tedious. So one question that remains is whether we can analyze this algorithm more elegantly.
%
%Also, this work focused on theoretical aspects. We are not aware of any empirical work using the idea of storing redundant transposed copies of data. We believe that it can positively help improve the cache efficiency of real systems. So, we hope to study this issue in a more empirical setting.

\section*{Acknowledgment}
We are grateful to Ning Chen, Edith Elkind, Sachin Lodha, Srinivasan Iyengar, Sasanka Roy and Dilys Thomas for useful discussions and ideas.

%\vspace{-.1in}
%\small
\bibliographystyle{plain}

\newpage
\section*{Appendix}

\begin{proof}[of Claim \ref{clm:general}]
Recall that, in the orthogonal CNN problem, any two consecutive pairs of requests must share either the same $x$ or $y$ coordinate. Given an online input sequence $I_{ortho} = (p_1, p_2, \ldots)$, we must construct an online sequence of requests $I_{cont}$ for the continuous CNN problem. When a new request $p_i$ arrives in $I_{ortho}$, we construct the next request in $I_{cont}$ as follows: $(p_{i-1}, \vd_{i-1})$, where $\vd_{i-1} = \frac{p_i - p_{i-1}}{|p_i - p_{i-1}|}$ is the unit vector in new direction. Clearly, the request trajectories in both the orthogonal and the continuous CNN problem instances are exactly the same. Therefore, any online algorithm $\mathcal{A}$ for $I_{cont}$ can also be used for $I_{ortho}$. \qed
\end{proof}
%
%\begin{proof}[of Claim \ref{clm:lb}]
%For this proof, we limit ourselves to input instances in which the requests can only appear in the vertices of the unit square with the origin at its bottom left.
For the purpose of this proof, we restrict requests to appear in an orthogonal manner, i.e., each request (except the first one) shares either the same $x$ or $y$ coordinate as the previous request.
%
%Assume that the server will not move once it is aligned to the request.
%Then, an adversary can force the online
%algorithm to move at each step. We keep server one step away from each
%request $e_i$ by placing next request $e_{i+1}$ diagonally opposite to
%the current position of the server.
%
%Given that sequence of requests, an offline algorithm splits it into
%consecutive triples. The algorithm can make at most one step to
%service each triple $e_1$, $e_2$ and $e_3$ just by hopping in one
%move to $e_2$ (if possible). Such a move is not possible only when the server and $e_2$ are diagonally opposite each other, so the server must be within 1 hop of both $e_1$ and $e_3$. It simply serves $e_1$ from the current position and then hops to $e_3$ to serve both $e_2$ and $e_3$.
%
%So, for every three steps of the online algorithm, there is an offline algorithm that makes at most 1 move, thereby completing the proof. \qed
%\end{proof}

\paragraph{Second tight instance for Remark~\ref{rem:tight}:} The second sequence is as follows. Both $\popt$ and $\pon$ be at $(0,1)$ in the beginning of the bishop phase. The request is at the origin. Then request  moves to $(1,0)$ and $\pon$  moves to $(1,0)$ and $\lon$ increases by $2$; $\popt$ moves to $(1,1)$. At this moment our algorithm is in the beginning of the rook phase with $\vo = -\mathbf{x}$. Request moves from $(1,0)$ to $(1,-\frac{1}{1+\sqrt{3}})$; OPT stays still. Server follows after the request untile $\vo = 0$. Thus we get to the bishop phase of the algorithm. Then request moves to $(1,1)$; $\popt$ and $\pon$ stay. Then, request  moves to  $(-\frac{1}{1+\sqrt{3}},1)$ and finally returns back to $(1,1)$; $\popt$ stays still; $\pon$ moves to $(1,1)$ at a  total length $3(1+\frac{1}{1+\sqrt{3}})$. At this point,  we are again at the beginning of the bishop phase. Then request moves to (say) $(1,0)$ and we are back to the starting situation, i.e., we are in the bishop phase with $\popt$ and $\pon$ coinciding  and are at unit distance away from the request. Hence, this cycle can be repeated ad infinitum. Clearly, $\frac{\lon}{\lopt} = 4(1+\frac{1}{1+\sqrt{3}})+1 = 3+2\sqrt{3}$.

\section*{The Unit CNN Problem}\label{sec:unit}
\begin{proof}[of Claim~\ref{clm:4c}]
We now provide a 4-competitive online algorithm. Our algorithm works in cycles. In each cycle, the online algorithm pays at most \$4 . Therefore, to prove that the algorithm is 4-competitive, we have to simply show that the offline optimal algorithm must pay \$1 per cycle. The intuition behind the algorithm is as follows. In each cycle, for the offline algorithm to avoid paying, there must be a {\em sweet spot} in $\R^2$ such that if the offline server were located there, it would not have to move throughout the cycle. The goal of the online algorithm is to discover such a position (if it exists) and reach it in at most four \$1 steps. The next cycle starts either (i) when the online algorithm establishes that there is no sweet spot (so the offline server has moved), or (ii) when the offline server cannot serve a request from the sweet spot (again, requiring the offline server to move).

Formally, the algorithm works as follows.
Assume that we are at the start of a cycle and the first request in that cycle, $r_1$, has arrived. We also assume that the offline algorithm has positioned its server in the advantageous sweet spot. The online algorithm pays \$1 (if required) and aligns with the $x_1$. We assume that $r_2$ does not share the same $x$ coordinate with $r_1$. If it does, the server need not move and it can be discarded from the input sequence (for analysis purposes). When $r_2 = (x_2, y_2)$ arrives, the online server moves to $(x_1, y_2)$ and serves $r_2$. If $y_1 = y_2$, then the sweet spot (if it exists) is somewhere on $y = y_1 = y_2$; this is case A. Otherwise, the sweet spot is either $(x_1, y_2)$ or $(x_2, y_1)$; this is case B.

Suppose we are in case A. Then we assume that $r_3 = (x_3, y_3)$ does not share the $y$ coordinate with $r_2$; if it did, the online server need not move and $r_3$ can be discarded for analysis purposes. The sweet spot must be $(x_3, y_1)$. The online algorithm can reach it in one step and stays there till a request that it cannot serve from $(x_3, y_1)$ arrives. The cycle is over.

Suppose we are in case B. If the third event $r_3 = (x_3, y_3)$ does not share either an $x$ or $y$ coordinate with $r_1$ or $r_2$, then clearly, there cannot be a sweet spot; the cycle is over. Suppose, instead, that $r_3$ shares the $x$ coordinate with $r_2$; other subcases can be seen symmetrically. Then, the sweet spot is $(x_2, y_1)$. Then, the online algorithm pays \$2 and moves to $(x_2, y_1)$. Again it stays there until forced to move when the cycle is over. Claim \ref{clm:4c} (stated in Section~\ref{sec:prelim}) follows in a straightforward manner.
Furthermore, there are instances for which the competitive ratio 4 can be realized, but we defer their description to the full version.
\qed\end{proof}
%\vspace{-0.05in}

%\vspace{-0.05in}

\begin{proof}[of Claim~\ref{clm:3c}]
Suppose the sequence of events is guaranteed to be orthogonal, i.e., each request shares a coordinate with the previous request. Then, we provide a very simple and tight algorithm with competitive ratio 3. The algorithm is very simple. The online server does not move unless the event is not visible to it. In that situation, it moves to the last event that it could see. By the orthogonality condition, we know that this algorithm is correct. Our claim that this algorithm is 3-competitive is along the same lines as the previous theorem, only simpler.

%\vspace{-0.2in}

We claim that for every consecutive sequence of moves worth $\$3$ in our algorithm, OPT has to do at least one \$1 move.
The proof is similar to that of Claim~\ref{clm:4c}. Consider four consecutive orthogonal requests $(r_0, r_1, r_2, r_3)$; note that adjacent requests must be at distinct locations. At the start of the cycle, both online and offline servers are in some position to serve $r_0$. We assume for the sake of contradiction that the offline algorithm has positioned itself so it does not have to move for the next three requests. The online server will move from its current position to $r_0 \to r_1 \to r_2$ to serve $r_1$, $r_2,$ and $r_3$, respectively. For the offline optimal algorithm to have served this sequence without moving, it must be in a position to see all four requests. The only candidate for the first three requests is $r_1$, while the only candidate for the second three requests is $r_2$. Since we require adjacent requests to be distinct, this is a contradiction. \qed
\end{proof}

%In the unit CNN problem, we have two metric spaces, $\M_r$ and $\M_c$. All distances between distinct elements are 1 in both metric spaces. Each metric space, $\M_r$ and $\M_c$, has a server, $s_r$ and $s_c$, respectively. An event is a pair $(r,c)$, where $r \in \M_r$ and $c \in \M_c$. An event is served if either $s_r$ moves to $r$ or $s_c$ moves to $c$. The input is an online sequence of events. When each event $(r,c)$ arrives, if neither server is in $r$ or $c$, then the algorithm must choose to move either $s_r$ to $r$ or $s_c$ to $c$.

%An alternative way to visualize this problem is as follows.

%Consider the sequence of $3$ points
%to which server come and the last event. If OPT did no moves, then
%it should seen all this sequence from certain place. But there is
%only one place from which OPT could seen first $3$ points, that is
%the second point. The same holds for the last $3$ points in our
%sequence. But these two places are different, since all four points
%we have considered are different.\qed
%\end{proof}

\end{document}